\documentclass[a4paper,reqno, 12pt]{amsart} 
\usepackage{latexsym,amsthm,amsmath,amssymb}
\usepackage[usenames,dvipsnames,svgnames,table]{xcolor}
\usepackage[utf8x]{inputenc}

\newtheorem{theorem}{Theorem}

\newtheorem{lemma}[theorem]{Lemma}

\theoremstyle{remark}

\usepackage{amssymb,amsmath,graphicx,enumitem,url,subcaption}
\usepackage[usenames,dvipsnames,svgnames,table]{xcolor}

\usepackage{algorithm,algorithmic}

\usepackage{tikz}
\usetikzlibrary{decorations.pathmorphing}
\usetikzlibrary{calc}
\tikzset{snake it/.style={decorate, decoration={snake, segment length=5pt, amplitude=1pt}}}

\colorlet{mylinkcolor}{violet}
\colorlet{mycitecolor}{YellowOrange}
\colorlet{myurlcolor}{Aquamarine}
\usepackage{hyperref}
\hypersetup{
    colorlinks,
    linkcolor={red!50!black},
    citecolor={blue!50!black},
    urlcolor={blue!80!black}
}

\newcommand{\calC}{\mathcal{C}}

\newcommand{\calH}{\mathcal{H}}

\newcommand{\R}{\mathbb{R}}
\newcommand{\DEF}{\em}
\newcommand{\cVD}{\textsc{Cluster-VD}}
\newcommand{\cVDapx}{\textsc{Cluster-VD-apx}}
\newcommand{\VC}{\textsc{Vertex Cover}}

\DeclareMathOperator\OPT{OPT}

\let\ge\geqslant
\let\leq\leqslant

\begin{document}

\title{Improved Approximation Algorithms for Hitting $3$-Vertex Paths}

\author[S.~Fiorini]{Samuel Fiorini}
\address[S.~Fiorini]{\newline D\'epartement de Math\'ematique
\newline Universit\'e Libre de Bruxelles
\newline Brussels, Belgium}
\email{sfiorini@ulb.ac.be}

\author[G.~Joret]{Gwena\"el Joret}
\address[G.~Joret]{\newline D\'epartement d'Informatique
\newline Universit\'e Libre de Bruxelles
\newline Brussels, Belgium}
\email{gjoret@ulb.ac.be}

\author[O.~Schaudt]{Oliver Schaudt}
\address[O.~Schaudt]{\newline Institut f\"ur Informatik 
\newline Universit\"at zu K\"oln
\newline K\"oln, Germany}
\email{schaudto@uni-koeln.de}

\thanks{S. Fiorini is supported by ERC Consolidator Grant 615640-ForEFront. G. Joret is supported by an ARC grant from the Wallonia-Brussels Federation of Belgium.}

\thanks{A preliminary version of this paper appeared as an extended abstract in the Proceedings of the 18th Conference on Integer Programming and Combinatorial Optimization (IPCO '16)~\cite{FJS16}.}

\maketitle

\begin{abstract}
We study the problem of deleting a minimum cost set of vertices from a given vertex-weighted graph in such a way that the resulting graph has no induced path on three vertices. 
This problem is often called {\em cluster vertex deletion} in the literature and admits a straightforward $3$-approximation algorithm since it is a special case of the vertex cover problem on a $3$-uniform hypergraph. 
Recently, You, Wang, and Cao described an efficient $5/2$-approximation algorithm for the unweighted version of the problem. 
Our main result is a $9/4$-approximation algorithm for arbitrary weights, using the local ratio technique.  
We further conjecture that the problem admits a $2$-approximation 
algorithm and give some support for the conjecture.    
This is in sharp contrast with the fact that the similar problem of deleting vertices to eliminate all triangles in a graph is known to be UGC-hard to approximate to within a ratio better than $3$, as proved by Guruswami and Lee.
\end{abstract}

\section{Introduction} \label{sec:intro}

Graphs in this paper are finite, simple, and undirected. 
Given a graph $G$ and cost function $c : V(G) \to \R_+$, the \emph{cluster vertex deletion problem} (\cVD) is to find a minimum cost set $X$ of vertices such that each component of $G - X$ is a complete graph. Equivalently, $X \subseteq V(G)$ is a feasible solution if and only if $G - X$ contains no induced subgraph isomorphic to $P_{3}$, the path on three vertices.

The problem admits a staightforward $3$-approximation algorithm: Assuming unit costs for simplicity, build any inclusionwise maximal collection $\calC$ of vertex-disjoint induced $P_3$'s in $G$ and include in $X$ every vertex covered by some member of $\calC$. If $\calC$ contains $k$ subgraphs then we get a lower bound of $k$ on the optimum. On the other hand, the cost of $X$ is $3k$.

The problem also admits an approximation-preserving reduction from \VC: if $H$ is any given graph, let $G$ denote the graph obtained from $H$ by adding a pendant edge to every vertex. Then solving \VC{} on $H$ is equivalent to solving \cVD{} on $G$. Hence, known hardness and inapproximability results for \VC{} apply to \cVD{} as well, and in particular it is UGC-hard to approximate \cVD{} to within any ratio better than $2$. We show that we can however come close to $2$.

\begin{theorem} \label{thm:main}
\cVD{} admits a $9/4$-approximation algorithm.
\end{theorem}

We further conjecture that \cVD{} can be $2$-approximated in polynomial time, 
as is the case for \VC{}. We give some support for this conjecture in 
Section~\ref{sec:conclusion}, where we notice that our $9/4$-approximation algorithms 
is in fact a $2$-approximation algorithm for the case where the largest clique in 
the input graph has size at most $4$, and can be easily modified to a 
$2$-approximation algorithm if the input graph does not contain any diamond 
($K_4$ minus an edge) as an induced subgraph. 

In contrast, the problem of finding a minimum cost set of vertices $X$ such that $G-X$ has no triangle is known to be UGC-hard to approximate to within 
any ratio better than $3$, as proved by Guruswami and Lee~\cite{GuruswamiLee14} 
(see also Guruswami and Lee~\cite{GuruswamiLee15} for related inapproximability results). 

\paragraph{\bf Previous Work.}

\cVD{} was previously mostly studied in terms of fixed parameter algorithms. 
H\"uffner, Komusiewicz, Moser, and Niedermeier~\cite{HuffnerEtAl} first gave an $O(2^k k^9 + nm)$-time 
fixed-parameter algorithm, parameterized by the solution size $k$, 
where $n$ and $m$ denote the number of vertices and edges of the graph, 
respectively. 
This was subsequently improved by Boral, Cygan, Kociumaka, and Pilipczuk~\cite{BoralEtAl}, who gave a 
$O(1.9102^k(n+m))$-time algorithm. 
See also Iwata and Oka~\cite{IwataOka} for related results 
in the fixed parameter setting. 

As for approximation algorithms, nothing better than a $3$-approximation 
was known until the recent work of You, Wang, and Cao~\cite{YouEtAl}, who 
showed that the unweighted version of \cVD{} admits a $5/2$-approximation algorithm. 

In a previous version of this paper~\cite{FJS16}, we gave a $7/3$-approximation algorithm for \cVD{}. The algorithm in this version of the paper achieves a better approximation ratio and is at the same time much simpler. 

Finally, we note that there has been recent activity on another restriction of the vertex cover problem on $3$-uniform hypergraph, namely, the feedback vertex set problem in tournaments. For that problem, the $5/2$-approximation algorithm by Cai, Deng and Zang~\cite{CaiDengZang01} was the best known for many years, until the very recent work of Mnich, Vassilevska Williams and V\'egh~\cite{MVV15} who found a $7/3$-approximation algorithm for the problem. 

\paragraph{\bf Our approach.} 

Our approximation algorithm is based on the {\em local ratio} technique. 
In order to illustrate the general approach, let us give a very simple $2$-approximation algorithm for hitting all $P_3$-{\em subgraphs} (instead of induced subgraphs) in a given weighted graph $(G,c)$, see Algorithm~\ref{algo_subgraph} below.

\begin{algorithm}\caption{$\textsc{Hitting-$P_3$-subgraphs-apx}(G,c)$}\label{algo_subgraph} 
\begin{algorithmic}
\REQUIRE $(G,c)$ a weighted graph
\ENSURE $X$ an inclusionwise minimal set of vertices hitting all the $P_3$-subgraphs
\IF{$G$ has no $P_3$ subgraph}
  \STATE{$X \leftarrow \varnothing$}
\ELSIF{$(G,c)$ has some zero-cost vertex $u$}
	\STATE{$X' \leftarrow \textsc{Hitting-$P_3$-subgraphs-apx}(G -u,c \text{ restricted to } V(G-u))$}
	\STATE{$X \leftarrow X'$ if $G - X'$ has no $P_3$-subgraph; $X \leftarrow X' \cup \{u\}$ otherwise} 
\ELSE
  \STATE{$u \leftarrow $ vertex of degree $d(u) \geqslant 2$, and let $(H,c_{H})$ be the weighted star centered}        
  \STATE{\quad at $u$ with $V(H) := N(u) \cup \{u\}$, $c_{H}(u) := d(u) - 1$ and $c_{H}(v) := 1$ for $v \in N(u)$}
  \STATE{$\lambda^* \leftarrow$ maximum scalar $\lambda$ s.t.\ $c(v) - \lambda c_H(v) \geqslant 0$ for all $v \in V(H)$}
  \STATE{$X \leftarrow \textsc{Hitting-$P_3$-subgraphs-apx}(G,c-\lambda^* c_H)$}
\ENDIF
\STATE{return $X$}
\end{algorithmic}
\end{algorithm}
 
It can be easily verified that the set $X$ returned by Algorithm~\ref{algo_subgraph} is an inclusionwise minimal feasible solution.
The reason why the algorithm is a $2$-approximation is that the optimum cost for the weighted star $(H,c_{H})$ is $d(u) - 1$ while the solution $X$ returned by the algorithm misses at least one of the vertices of the star, and thus has a local cost of at most $2(d(u)-1)$.

We remark that a $2$-approximation algorithm for the problem of hitting $P_3$-subgraphs 
can also be obtained via 
a straightforward modification of the primal/dual $2$-approximation algorithm 
of Chudak {\it et al.}~\cite{chudak} for the feedback vertex set problem. 
(Indeed, this is exactly what was done by Tu and Zhou~\cite{tu}.) 
However, the resulting algorithm is much more complicated than Algorithm~\ref{algo_subgraph}.

It is perhaps worth pointing out that, in the case of triangle-free graphs, 
hitting $P_3$'s or induced $P_3$'s are the same problem. This was actually 
an important insight for the $5/2$-approximation algorithm of You, Wang, and Cao~\cite{YouEtAl}. However, for arbitrary graphs the induced version of the problem seems much more difficult. Nevertheless, we are tempted to take the simplicity of 
Algorithm~\ref{algo_subgraph} as a hint that the local ratio technique 
is a good approach to attack the problem. 

From a high level point of view, the structure of our $9/4$-approximation algorithm for \cVD{} 
is as follows: 
As long as there is an induced $P_3$ in the graph, 
either we can apply a reduction operation (identifying {\em true twins}) 
that does not change the optimum, or we find some induced subgraph $H$ 
and decrease the weights of its vertices in $(G,c)$ proportionally to a carefully 
chosen weighting $c_H$ for the vertices of $H$, ensuring a local ratio of $9/4$. 
(We remark that $c_H$ depends on $H$ only and is thus independent of the weights of vertices in $G$, similarly as in Algorithm~\ref{algo_subgraph}.)

The induced subgraphs we consider are as follows: \emph{cycles of length~$4$} 
($C_4$'s), \emph{$5$-cliques plus distinguishing sets} ($K_5$'s plus distinguishing 
sets), and \emph{second-neighborhood} subgraphs induced by the vertices at
distance at most two from a maximum degree vertex of $G$. 
We note that the approximation algorithm in the preliminary version of this paper~\cite{FJS16} has the same general structure but exploits a different set of induced subgraphs, namely a finite (but longish) list of graphs on at most $7$ vertices. 
Using the new set of induced subgraphs results in both simpler proofs and a better approximation ratio of $9/4$.

\section{Definitions and Preliminaries} \label{sec:def_prelim}

Let $G$ be a graph. Recall that the feasible solutions to \cVD{} in $G$ are the sets of vertices $X$ that intersect every induced subgraph isomorphic to $P_3$. For this reason, we call such sets $X$ {\DEF hitting sets} of $G$. We denote by $\OPT(G)$ the minimum size of a hitting set of $G$. The definitions extend naturally in the weighted setting: Given a weighted graph $(G, c)$, where $c: V(G) \to \R_+$, we let $\OPT(G, c)$ denote the minimum weight (cost) of a hitting set of $G$. As expected, the \emph{weight} (or \emph{cost}) of set $X \subseteq V(G)$ is defined as $c(X) := \sum_{v \in X} c(v)$.

For $X \subseteq V(G)$, the subgraph of $G$ induced by $X$ is denoted by $G[X]$. When $H$ is an induced subgraph of $G$ or isomorphic to an induced subgraph of $G$, we sometimes say that \emph{$G$ contains $H$}. If $G$ does not contain $H$, we also say that $G$ is \emph{$H$-free}.

For $v \in V(G)$, the neighborhood of $v$ is denoted by $N(v)$. From time to time, to indicate that $x$ is a neighbor of $y$, we simply say that \emph{$x$ sees $y$}.

\section{Tools} \label{sec:tools}

\subsection{True Twins and Distinguishers} \label{sec:true_twins}

Two vertices $u, u'$ of a graph $G$ are called {\DEF true twins} if they are adjacent and have the same neighborhood in $G - \{u,u'\}$. True twins have a particularly nice behavior regarding \cVD, as proved in our next lemma. This is our first main technical tool.

\begin{lemma} \label{lem:true_twins}
Let $(G,c)$ be a weighted graph and $u, u' \in V(G)$ be true twins. Let $(G',c')$ denote the weighted graph obtained from $G$ by transferring the whole cost of $u'$ to $u$ and then deleting $u'$, that is, let $G' := G - u'$ and $c'(v) := c(v)$ if $v \in V(G'), v \neq u$ and $c'(v) := c(u) + c(u')$ if $v = u$. Then $\OPT(G,c) = \OPT(G',c')$. 
\end{lemma}

\begin{proof}
We have $\OPT(G,c) \leqslant \OPT(G',c')$ because every hitting set $X'$ of $G'$ yields a hitting set $X$ of $G$ with the same cost: we let $X := X' \cup \{u'\}$ if $X$ contains $u$ and $X := X'$ otherwise. 
Here we use that no induced $P_3$ in $G$ contains both $u$ and $u'$.

Conversely, we have $\OPT(G',c') \leqslant \OPT(G,c)$ because any inclusionwise minimal cost hitting set $X$ of $G$ either contains both of the true twins $u$ and $u'$, or none of them.
\end{proof}

If $G$ does not contain any pair of true twins, we say that $G$ is \emph{twin-free}.

Notice that two adjacent vertices $u$ and $v$ are \emph{not} true twins if and only if $G$ has an induced $P_3$ containing $u$ and $v$. The third vertex of such a $P_3$ is adjacent to one of $u$ and $v$, and nonadjacent to the other. We say that it is a \emph{distinguisher} for the edge $uv$, and call the induced $P_3$ a \emph{distinguishing} $P_3$.

Now let $S \subseteq V(G)$. A set $D \subseteq V(G)$ disjoint from $S$ is said to be a \emph{distinguishing set} for $S$ if for every edge $uv$ whose endpoints are true twins in $G[S]$, the set $D$ contains a distinguisher $w$ for the edge $uv$. 
See Figure~\ref{fig:distinguishing_set} for an illustration. 

\begin{figure}
\centering
\includegraphics[width=0.14\textwidth]{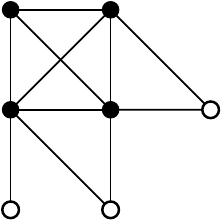}
\caption{The white vertices form a distinguishing set for the black vertices.\label{fig:distinguishing_set}} 
\end{figure}

\begin{lemma} \label{lem:Oli's_lemma}
Let $H$ be a graph whose vertex set is partitioned into a clique $C$ and a distinguishing set $D$ for $C$. Then, there exists a weight function $c_H : V(H) \to \mathbb{Z}_{\geqslant 0}$ such that $c_H(v) = 1$ for all $v \in C$, $\sum_{v \in D} c_H(v) = |C| - 1$ and every set $X \subseteq V(H)$ hitting each distinguishing $P_3$ has weight $c_H(X) \geqslant |C| - 1$. In particular, $\OPT(H,c_H) \geqslant |C| - 1$.
\end{lemma}

\begin{proof}
First, we claim that for every fixed $w \in D$, the set of edges $uv$ of $C$ that are distinguished by $w$ and by no other vertex of $D$ forms a matching. Indeed, assume that $C$ has two incident edges $uv$ and $uv'$ that are distinguished by $w$ but are not distinguished by any other vertex of $D$. 
Then either $w$ is adjacent to both $v$ and $v'$, or to none of them. 
Thus $w$ does not distinguish the edge $vv'$. Let $w' \in D$ be any vertex distinguishing the edge $vv'$. Then $w'$ is a distinguisher of $uv$ or $uv'$ that is distinct from $w$, a contradiction.

Next, we define the weight function $c_H$ by the following iterative procedure.

\begin{itemize}
\item Pick any distinguisher $w \in D$.
\item Let $M$ denote the edges of $C$ that are distinguished by $w$ and by no other vertex of $D$. By the claim, $M$ is a matching. Define $c_H(w) := |M|$.  
\item Let $U$ be any set of $|M|$ vertices hitting each edge of $M$ exactly once. Delete the vertices of $U$ from $C$, delete $w$ from $D$, and repeat until there are no more vertices in $D$.
\end{itemize}

Notice that at each step $D$ remains a distinguishing set for $C$. Notice also that the graph obtained after deleting $w$ from the distinguishing set and $U$ from the clique does not depend on the particular choice of $U$. Indeed, all the possible choices for $U$ lead to isomorphic graphs since $w$ gets deleted.

Finally, we show that every set $X \subseteq V(H)$ hitting all the distinguishing $P_3$'s has weight at least $|C| - 1$, by induction. 

Let $w$ denote the first distinguisher picked by the weighting procedure and the corresponding set $U$. If $w  \in X$, consider the reduced instance $C' := C \setminus U$, $D' := D \setminus \{w\}$. It is true that $X - w$ hits all the distinguishing $P_3$'s for this new instance. By induction, we get $c_H(X) = c_H(X \setminus \{w\}) + c_H(w) \ge |C'| - 1 + c_H(w) = |C| - c_H(w) - 1 + c_H(w) = |C| - 1$. 

Now assume that $w \notin X$. Thus $X$ meets each edge of $M$ at least once. Let $R \subseteq X$ be any set meeting each edge of $M$ exactly once. By the remark above, we may assume that $U = R$. As before, consider the reduced instance $C' := C \setminus U$, $D' := D \setminus \{w\}$. Clearly, $X \setminus U$ hits all the distinguishing $P_3$'s for this instance. By induction, we get $c_H(X) = c_H(X \setminus U) + c_H(U) = c_H(X \setminus U) + c_H(w) \ge |C'| - 1 + c_H(w) = |C| - c_H(w) - 1 + c_H(w) = |C| - 1$.
\end{proof}

\subsection{$\alpha$-Good Induced Subgraphs} \label{sec:good_H}

Given a graph $G$, an induced subgraph $H$ of $G$, and a weighting $c_{H} : V(H) \to \mathbb{R}_+$, we say that $(H, c_{H})$ is {\DEF $\alpha$-good in $G$} if for every inclusionwise minimal hitting set $X$ of $G$ we have
\begin{equation} \label{eq:alpha_good}
\sum_{v\in X\cap V(H)} c_{H}(v) \leqslant \alpha \cdot \OPT(H, c_H)\,.
\end{equation}
Moreover, we say that an induced subgraph $H$ of $G$ is itself {\DEF $\alpha$-good in $G$} if there exists a weighting $c_{H}$ such that $(H,c_{H})$ is $\alpha$-good. 

We start by considering two different types of weighted induced subgraphs $(H,c_{H})$ that satisfy the stronger condition $\sum_{v \in V(H)} c_{H}(v) \leqslant \alpha \cdot \OPT(H, c_H)$, which obviously implies that they are $\alpha$-good.

\begin{lemma} \label{lem:C_4}
Let $G$ be a graph. If $H$ is an induced $C_4$ in $G$, then $H$ is $2$-good.
\end{lemma}
\begin{proof}
We let $c_H(v) := 1$ for all $v \in V(H)$. Then $\OPT(H,c_H) = 2$ and
$$
\sum_{v \in V(H)} c_H(v) = 4 = 2 \cdot \OPT(H,c_H)\,.
$$
\end{proof}

\begin{lemma} \label{lem:5-clique}
Let $G$ be a twin-free graph, let $C$ be a $5$-clique in $G$ and let $D$ be a distinguishing set for $C$. The induced subgraph $H := G[C \cup D]$ is $\alpha$-good in $G$ for $\alpha = 9/4$.
\end{lemma}

\begin{proof}
With the weight function $c_H$ defined in Lemma~\ref{lem:Oli's_lemma}, we have
$$
\sum_{v \in V(H)} c_H(v) 
= |C| + |C| - 1
= 9 
\leqslant (9/4) \cdot \OPT(H,c_H)\,.
$$
\end{proof}

The next lemma is our main tool for constructing $\alpha$-good weighted induced
subgraphs for $\alpha = 2$. This time, we use the minimality of the hitting set $X$
to establish $\alpha$-goodness, however in a very simple way.

\begin{lemma} \label{lem:2nd_neighborhood}
Let $G$ be a graph that is twin-free, $C_4$-free and $K_5$-free. 
Let $v_0$ be a vertex of maximum degree, and let $A_1$, \ldots, $A_k$ denote the components of $G[N(v_0)]$.  For $i \in [k]$, let $B_i$ denote the set of vertices in $G - (\{v_0\} \cup N(v_0))$ that see at least one vertex in $A_i$. Let $H$ denote the subgraph of $G$ induced by $\{v_0\} \cup N(v_0) \cup \bigcup_{i=1}^k B_i$. Then there exists a weight function $c_H : V(H) \to \mathbb{Z}_{\geqslant 0}$ such that $(H,c_{H})$ is $2$-good in $G$.
\end{lemma}

\begin{proof}
Notice that since $G$ is $C_4$-free, the sets $B_i$ are pairwise disjoint. 

In all cases except in one sporadic case (part of Case 1.3 below), we let $c_H(v) := 1$ for all $v \in N(v_0)$, that is, we put unit weight on these vertices. The weights on the vertices in $\{v_0\} \cup \bigcup_{i=1}^k B_i$ will be determined later. 

Let $X$ denote a minimal hitting set of $G$. We wish to show that \eqref{eq:alpha_good} always holds for our choice of weights and $\alpha = 2$. We split the discussion into two cases according to the number of components of $G[N(v_0)]$. Each of these cases is split into several subcases according to the structure of the induced subgraphs $G[A_i]$, $i \in [k]$.

In all the cases, we make sure that the weight on $v_0$ is at least $1$, and hence
$$
\sum_{v \in X \cap V(H)} c_H(v) \leqslant \sum_{v \in V(H)} c_H(v) - 1\,.
$$
This follows from the assumption that $X$ is minimal: $X$ has to exclude at least one of the vertices of $\{v_0\} \cup N(v_0)$, and each of these vertices has weight at least~$1$. In order to prove $2$-goodness, it suffices then to show that $c_H(V(H)) \leqslant 2 \OPT(H,c_H) + 1$.\medskip

\noindent \emph{Case 1. $k = 1$.} 
Then $A_1 = N(v_0)$.\medskip

\noindent \emph{Case 1.1. $A_1$ is a clique.} We let $c_H(v_0) := 1$ and use Lemma~\ref{lem:Oli's_lemma} on the clique $C = \{v_0\} \cup A_1$ and distinguishing set $D = B_1$ to define weights on $B_1$. We get $c_H(V(H)) = 2|C| - 1$ and $\OPT(H,c_H) \geqslant |C| - 1$, and thus $c_H(V(H)) \leqslant 2 \OPT(H,c_H) + 1$.

\medskip

\noindent \emph{Case 1.2. $A_1$ is not a clique and $G[A_1]$ has clique number $2$.} If $|A_1| \geqslant 4$, we let $c_H(v_0) := |A_1| - 3 \geqslant 1$ and $c_H(v) := 0$ for $v \in B_1$. Then $\OPT(H,c_H) \geqslant |A_1| - 2$. This can be seen as follows. Let $Y$ denote a minimum weight hitting set of $(H,c_H)$. Either $Y$ contains $v_0$ and at least one vertex of $A_1$, or $Y$ does not contain $v_0$ and $A_1 \setminus Y$ is a clique. In both cases the weight of $Y$ is at least $|A_1| - 2$. We have $c_H(V(H)) = 2|A_1| - 3 \leqslant 2 \OPT(H,c_H) + 1$.

Otherwise, $|A_1| = 3$ and $G[A_1]$ is a $P_3$. Let $v_1$ denote the middle vertex of this $P_3$. Since $G$ is twin-free, $v_0$ and $v_1$ are not true twins. Thus, there exists a vertex $v_2 \in B_1$ that sees $v_1$ and not $v_0$. We put unit weights on $v_0$ and $v_2$, and zero weights on the vertices of $B_1 \setminus \{v_2\}$. We get $c_H(V(H)) = 5 \leqslant 2 \OPT(H,c_H) + 1$.
\medskip

\noindent \emph{Case 1.3. $A_1$ is not a clique and $G[A_1]$ has clique number $3$.} First, assume that $|A_1| \geqslant 6$ and the minimum size of a hitting set of $G[A_1]$ is at least $2$. We let $c_H(v_0) := |A_1| - 5 \geqslant 1$ and $c_H(v) := 0$ for $v \in B_1$. By an argument similar to that used in Case 1.2, we have $\OPT(H,c_H) \geqslant |A_1| - 3$. Then $c_H(v(H)) = 2|A_1| - 5 \leqslant 2 \OPT(H,c_H) + 1$.

Second, assume that there is a vertex $v_1$ that is a hitting set of $G[A_1]$. 
Because $G[A_1]$ is connected, not a clique, and does not contain any $4$-clique, one can check that the following holds for the graph $G[A_1]$: (1) $v_1$ has no true twin, (2) every pair of true twins lie in a triangle, (3) every triangle contains a pair of true twins, and (4) every two pairs of true twins are vertex-disjoint and there is no edge between them.  

There is at least one pair of true twins in $G[A_1]$ (since $G[A_1]$ has a triangle), and each pair of true twins in $G[A_1]$ is distinguished in $G$ by some vertex in $B_1$. 
Moreover, every two such pairs are distinguished by distinct vertices in $B_1$, since $G$ is $C_4$-free. 
So there is a nonempty set $B'_1 \subseteq B_1$ with the following properties: (i) every pair of true twins in $G[A_1]$ has a distinguisher in $B'_1$, (ii) there are $|B'_1|$ vertex-disjoint induced $P_3$'s with one endvertex in $B'_1$ and the other two vertices in $A_1 \setminus \{v_1\}$.

Assume for now that $|A_1| - |B'_1| - 3 \geqslant 1$. Then, we put a weight of $|A_1| - |B'_1| - 3$ on $v_0$, unit weights on the vertices of $B'_1$ and zero weights on the vertices of $B_1 \setminus B'_1$. Consider a minimum weight hitting set $Y$ of $(H,c_H)$. Either $Y$ contains $v_0$ and at least $1 + |B'_1|$ further vertices in $A_1 \cup B'_1$, or $Y$ does not contain $v_0$ and contains at least $|A_1| - 2$ vertices in $A_1 \cup B'_1$. Therefore, we have $\OPT(H,c_H) \geqslant |A_1| - 2$ and $c_H(V(H)) = 2|A_1| - 3 \leqslant 2 \OPT(H,c_H) + 1$.

Otherwise, $|A_1| - |B'_1| - 3 \leqslant 0$ and using $|A_1| \geqslant 4$, $|B'_1| \geqslant 1$ and $|A_1| \geqslant 2|B'_1| + 1$, we have $(|A_1|,|B'_1|) \in \{(4,1),(5,2)\}$. In both cases, $A_1$ contains a vertex $v_2$ (possibly $v_2 = v_1$) that sees every vertex in $A_1 \setminus \{v_2\}$, in addition to $v_0$. Since $v_0$ has maximum degree, $v_2$ has exactly the same neighbors as $v_0$, and is thus a true twin of $v_0$, a contradiction.

Finally, the last case to consider is when the minimum size of a hitting set in $G[A_1]$ is at least $2$ and $|A_1| \leq 5$. 
Using that $G[A_1]$ is $C_4$-free, one can check that $|A_1| = 5$ in this case, and that the minimum size of a hitting set in $G[A_1]$ is exactly $2$. 
Then the maximum degree in $G[A_1]$ is at most $3$ since otherwise by maximality of its degree, $v_0$ would have a true twin in $A_1$. Since $G[A_1]$ contains at least one triangle and is $C_4$-free, this leaves only one possibility: $G[A_1]$ is a \emph{bull}, that is, a triangle with two extra vertices of degree~$1$, say $v_1$ and $v_2$, each seeing a different vertex in the triangle. We increase the weight of one of these vertices to $2$, say $v_1$, put a unit weight on $v_0$, an zero weights on $B_1$. Then $\OPT(H,c_H) = 3$ and $c_H(V(H)) = 7 \leqslant 2 \OPT(H,c_H) + 1$.\medskip

\noindent \emph{Case 2. $k \geqslant 2$.} In this case the weight on $v_0$ is set implicitly. Remember that we require
\begin{align}
\label{eq:c_H(v_0)_positive} c_H(v_0) &\geqslant 1\,.
\end{align}

For $i \in [k]$, we let $\OPT_i$ denote the minimum weight of a hitting set of $H[A_i \cup B_i]$, and $\OPT'_i$ denote the minimum weight of a hitting set of $H[\{v_0\} \cup A_i \cup B_i]$ not containing $v_0$. Notice that these quantities depend on the weight function $c_H$, which is not fully determined at this point. 

We claim that the following lower bound holds on $\OPT(H,c_H)$, regardless of how $c_H(v)$ is chosen for $v \in \{v_0\} \cup \bigcup_{i=1}^k B_i$:
$$
\OPT(H,c_H) \geqslant \min \left(\left\{c_H(v_0) + \sum_{i} \OPT_i \right\} \cup \left\{\sum_{i \neq j} |A_i| + \OPT'_j \mid j \in [k]\right\}\right)\,.
$$

In order to verify that this claim is true, consider a hitting set $Y$ of $H$. If $Y$ contains $v_0$, then $Y \cap (A_i \cup B_i)$ is a hitting set of $G[A_i \cup B_i]$ for each $i \in [k]$. In this case, $c_H(Y) \geqslant c_H(v_0) + \sum_{i} \OPT_i$. Otherwise, $Y$ does not contain $v_0$. Then, there exists an index $j \in [k]$ such that $Y$ contains $A_i$ for all $i \neq j$. Moreover, $Y \cap (A_j \cup B_j)$ is a hitting set of $G[\{v_0\} \cup A_i \cup B_i]$ not containing $v_0$. In this case, $c_H(Y) \geqslant \sum_{i \neq j} |A_i| + \OPT'_j$. 

Thanks to the above lower bound on $\OPT(H,c_H)$, it suffices to satisfy the following $1 + k$ inequalities in order to guarantee that $(H,c_H)$ is $2$-good (remember that we put unit weights over the $A_i$'s, thus $c_H(A_i) = |A_i|$ for every $i \in [k]$):
\begin{align}
\nonumber c_H(v_0) + \sum_{i} (|A_i| + c_H(B_i))
&\leqslant 2 \left(c_H(v_0) + \sum_{i} \OPT_i \right) + 1\\
\label{eq:LB_c_H(v_0)}\iff c_H(v_0) &\geqslant \sum_{i} \left(|A_i| + c_H(B_i) - 2 \, \OPT_i \right) - 1
\end{align}
and, for all $j \in [k]$,
\begin{align}
\nonumber c_H(v_0) + \sum_{i} (|A_i| + c_H(B_i))
&\leqslant 2 \left(\sum_{i \neq j} c_H(A_i) + \OPT'_j \right) + 1\\
\label{eq:UB_c_H(v_0)}\iff c_H(v_0) &\leqslant \sum_{i \neq j} (|A_i| - c_H(B_i)) + 2\,\OPT'_j - |A_j| - c_H(B_j) + 1\,.
\end{align}
By eliminating the variable $c_H(v_0)$ from the system \eqref{eq:c_H(v_0)_positive}--\eqref{eq:UB_c_H(v_0)}, we get the following $2k$ inequalities not involving $c_H(v_0)$. For all $j \in [k]$:
\begin{align}
\label{eq:one}
|A_j| + c_H(B_j) &\leqslant \sum_{i \neq j} (\OPT_i - c_H(B_i)) + \OPT_j + \OPT'_j + 1
\end{align}
and
\begin{align}
\label{eq:two}
|A_j| + c_H(B_j) &\leqslant \sum_{i \neq j} (|A_i| - c_H(B_i)) +  2\,\OPT'_j\,.
\end{align}
If \eqref{eq:one} and \eqref{eq:two} are satisfied for all $j \in [k]$, then $(H,c_H)$ is $2$-good.\medskip

In order to simplify these constraints, we add the extra requirements that $c_H(B_i) \leqslant \OPT_i$ and $c_H(B_i) \leqslant |A_i| - 1$ for all $i \in [k]$. Since $k \geqslant 2$ and $\OPT'_j \geqslant \OPT_j$, both \eqref{eq:one} and \eqref{eq:two} follow if, for all $j \in [k]$:
\begin{equation}
\label{eq:target}
|A_j| + c_H(B_j) \leqslant 1 + \OPT_j +\OPT'_j\,.
\end{equation}
Fix any $j \in [k]$. We set the weights on the vertices of $B_j$ by inspecting the structure of the induced graph $H[A_j]$. We consider three subcases, see below. In each of these cases, it is straightforward to check that the two extra requirements are satisfied for $i = j$.\medskip

\noindent \emph{Case 2.1. $A_j$ is a clique.} By Lemma~\ref{lem:Oli's_lemma}, we may set weights on $B_j$ to have $c_H(B_j) = |A_j| - 1$ and $\OPT_j = |A_j| - 1$. Now $\OPT'_j \geqslant \OPT_j = |A_j| - 1$, so that inequality \eqref{eq:target} is satisfied, since
$$
|A_j| + c_H(B_j) = |A_j| + |A_j| - 1 = 1 + (|A_j| - 1) + (|A_j| - 1) 
\leqslant 1 + \OPT_j +\OPT'_j\,.
$$

\noindent \emph{Case 2.2. $A_j$ is not a clique and $G[A_j]$ has clique number $2$.} In this case, we put zero costs on $B_j$. We get $\OPT_j \geqslant 1$ because $A_j$ is not a clique and also $\OPT'_j \geqslant |A_j| - 2$, so that
$$
|A_j| + c_H(B_j) = |A_j| = 1 + 1 + (|A_j| - 2) \leqslant 1 + \OPT_j +\OPT'_j\,.
$$

\noindent \emph{Case 2.3. $A_j$ is not a clique and $G[A_j]$ has clique number $3$.} If the minimum size of a hitting set in $G[A_j]$ is at least $2$, we put zero weights on $B_j$. Thus \eqref{eq:target} is satisfied, since then $\OPT_j \geqslant 2$ and
$$
|A_j| + c_H(B_j) = |A_j| \leqslant 1 + 2 + |A_j| - 3
\leqslant 1 + \OPT_j + \OPT'_j\,.
$$

Now, assume that there exists some vertex $v_1$ that hits all the induced $3$-paths in $A_j$. As in Case 1.3, we see that there is a set $B'_j \subseteq B_j$ with the following properties: (i) every pair of true twins in $G[A_j]$ has a distinguisher in $B'_j$, (ii) among the distinguishing $P_3$'s defined by the vertices in $B'_j$, there are $|B'_j|$ vertex-disjoint $P_3$'s. 

We put unit weights on the vertices of $B'_j$ and zero weight on the vertices of $B_j \setminus B'_j$. We get $\OPT_j \geqslant |B'_j| + 1$ since a hitting set in $G[A_j \cup B_j]$ has to have one vertex on each of the $|B'_j|$ vertex-disjoint distinguishing $P_3$'s but this is not enough to hit all the induced $P_3$'s. And also $\OPT'_j \geqslant |A_j| - 2$ since every triangle in $G[A_j]$ has one pair of true twins, which is distinguished by some vertex of $B'_j$. We have
$$
|A_j| + c_H(B_j) = |A_j| + |B'_j| \leqslant 1 + (|A_j| - 2) + (|B'_j| + 1)
\leqslant 1 + \OPT_j + \OPT'_j\,.
$$
\end{proof}

\section{Algorithm} \label{sec:algo}

Our $9/4$-approximation algorithm is described below, see Algorithm~\ref{algo}. Although we could have presented it as a primal-dual algorithm, we chose to present it within the local ratio framework in order to avoid some technicalities, especially those related to the elimination of true twins. 

The following lemma makes explicit a simple property of \cVD{} 
that is key when using the local ratio technique. 
This property is common to many minimization problems, and is 
often referred to as the {\em Local Ratio Lemma}; see e.g.\ the survey of 
Bar-Yehuda, Bendel, Freund, and Rawitz~\cite{Bar-YehudaEtAlSurvey}. 

\begin{lemma}[Local Ratio Lemma for \cVD{}]
\label{lem:localratiolemma} 
Let $(G, c)$ be a weighted graph with $c$ the sum of two cost functions $c'$ and $c''$, 
and let $\alpha \geqslant 1$. 
If $X$ is a hitting set of $G$
such that $c'(X) \leqslant \alpha \cdot \OPT(G, c')$ 
and $c''(X) \leqslant \alpha \cdot \OPT(G, c'')$, then 
$c(X)  \leqslant \alpha \cdot \OPT(G, c)$. 
\end{lemma}
\begin{proof} 
Since $c(X)=c'(X) + c''(X)$, it is enough to show that $\OPT(G, c') + \OPT(G, c'') \leqslant \OPT(G, c)$. 
To see this, let $X^*$ be a minimum weight hitting set for $(G, c)$. 
Then $\OPT(G,c) = c(X^*) = c'(X^*) + c''(X^*) \geqslant \OPT(G, c') + \OPT(G, c'')$. 
\end{proof}

\begin{algorithm}\caption{$\cVDapx(G,c)$}\label{algo} 
\begin{algorithmic}[1]
\REQUIRE $(G,c)$ a weighted graph
\ENSURE $X$ an inclusionwise minimal hitting set of $G$
\IF{$G$ is a disjoint union of cliques}
	\STATE{$X \leftarrow \varnothing$}
\ELSIF{there exists $u \in V(G)$ with $c(u) = 0$}
	\STATE{$G' \leftarrow G - u$}
	\STATE{$c'(v) \leftarrow c(v)$ for $v \in V(G')$}
	\STATE{$X' \leftarrow \cVDapx(G',c')$ \label{step:zero_cost}}
	\STATE{$X \leftarrow X'$ if $X'$ is a hitting set of $G$; $X \leftarrow X' \cup \{u\}$ otherwise}
\ELSIF{there exist true twins $u, u' \in V(G)$}
	\STATE{$G' \leftarrow G - u'$}
	\STATE{$c'(v) \leftarrow c(u) + c(u')$ for $v = u$; $c'(v) \leftarrow c(v)$ for $v \in V(G') \setminus \{u\}$}
	\STATE{$X' \leftarrow \cVDapx(G',c')$ \label{step:true_twins}}
	\STATE{$X \leftarrow X'$ if $X'$ does not contain $u$; $X \leftarrow X' \cup \{u'\}$ otherwise}
\ELSE
	\STATE{pick the first $(H,c_H)$ in $\calH(G)$ \label{step:find_good_H}}
	\STATE{$\lambda^* \leftarrow \max \{\lambda \mid \forall v \in V(H) : c(v) - \lambda c_H(v) \geqslant 0\}$}
	\STATE{$G' \leftarrow G$}
	\STATE{$c'(v) \leftarrow c(v) - \lambda^* c_H(v)$ for $v \in V(H)$; $c'(v) \leftarrow c(v)$ for $v \in V(G) \setminus V(H)$}
	\STATE{$X \leftarrow \cVDapx(G',c')$ \label{step:increase_dual}}
\ENDIF
\STATE{return $X$}
\end{algorithmic}
\end{algorithm}

Algorithm~\ref{algo} uses an ordered list $\calH(G)$ of weighted induced subgraphs $(H,c_H)$ of $G$ as defined in Lemmas~\ref{lem:C_4}, \ref{lem:5-clique} and \ref{lem:2nd_neighborhood}. We order the weighted induced subgraphs $(H,c_H)$ in $\calH(G)$ in order to make sure that the hypotheses of the corresponding lemma are satisfied when $(H,c_H)$ is used. The first elements of the list are induced $C_4$'s (if any), next come the induced $K_5$'s (if any) each of them taken together with a distinguishing set, and finally the second neighborhood of any maximum degree vertex $v_0$. Notice that the list $\calH(G)$ is always nonempty and of polynomial size. This ensures that Algorithm~\ref{algo} has polynomial complexity.

We are now ready to prove our main result.

\begin{proof}[Proof of Theorem~\ref{thm:main}]
By induction on the number of recursive calls, we prove the following claim:
\begin{quote} \it
$(\star)$ The set $X$ output by Algorithm~\ref{algo} on input $(G,c)$ is an inclusionwise minimal hitting set of $G$ and $c(X) \leqslant \frac{9}{4} \cdot \OPT(G,c)$.
\end{quote}
If the algorithm does not call itself, then it returns the empty set and in this case claim $(\star)$ trivially holds. Now assume that the algorithm calls itself at least once and that the output $X'$ of the recursive call is an inclusionwise minimal hitting set of $G'$ that satisfies $c'(X') \leqslant \frac{9}{4} \cdot \OPT(G',c')$. There are three cases to consider.\smallskip

\noindent \emph{Case 1:} The recursive call occurs at Step~\ref{step:zero_cost}. Then we have $c(X) = c'(X')$ and $\OPT(G,c) = \OPT(G',c')$ because $(G',c')$ is simply $(G,c)$ with one zero-cost vertex removed. By construction, $X$ is an inclusionwise minimal hitting set of $G$. Moreover, by what precedes, $c(X) = c'(X') \leqslant \frac{9}{4} \cdot \OPT(G',c') = \frac{9}{4} \cdot \OPT(G,c)$.\smallskip

\noindent \emph{Case 2:} The recursive call occurs at Step~\ref{step:true_twins}. Again, $X$ is an inclusionwise minimal hitting set of $G$ and $c(X) = c'(X') \leqslant \frac{9}{4} \cdot \OPT(G',c') = \frac{9}{4} \cdot \OPT(G,c)$, where the last equality holds by Lemma~\ref{lem:true_twins}.\smallskip

\noindent \emph{Case 3:} The recursive call occurs at Step~\ref{step:increase_dual}. In this case, $G = G'$ and $X = X'$, thus $X$ is automatically an inclusionwise minimal hitting set of $G$. Let $c''$ denote the weighting $c_H$ extended to $V(G)$ by letting $c''(v) := 0$ for $v \in V(G) \setminus V(H)$. We have $c'(X) \leqslant \frac{9}{4} \cdot \OPT(G,c')$ by induction and $\lambda^* c''(X) \leqslant \frac{9}{4} \cdot \OPT(G,\lambda^* c'')$ since all the weighted induced subgraphs $(H,c_H)$ in $\calH(G)$ are $9/4$-good in $G$ (see Lemmas~\ref{lem:C_4}, \ref{lem:5-clique} and \ref{lem:2nd_neighborhood}). Because $c = c' + \lambda^* c''$, Lemma~\ref{lem:localratiolemma} implies $c(X) \leqslant \frac{9}{4} \cdot \OPT(G,c)$.
\end{proof}

\section{Conclusion} \label{sec:conclusion}

In this paper we presented a $9/4$-approximation algorithm for the \cVD{} problem, based on the local ratio technique. The main idea underlying the algorithm is that in a twin-free, $(C_4,K_5)$-free graph, one can define weights on the vertices of the second neighborhood of any maximum degree vertex in order to guarantee a local ratio of at most $2$. Moreover, the input graph can be made twin-free and $C_4$-free without worsening the approximation ratio beyond~$2$. Making the graph $K_5$-free is what causes the approximation ratio to increase to~$9/4$. If the input graph is $K_5$-free, our algorithm is in fact a $2$-approximation algorithm.

Furthermore, looking closely at the proof of Lemma~\ref{lem:2nd_neighborhood}, we see that one can also obtain a $2$-approximation algorithm for diamond-free graphs. This is due to the fact that, if $G$ is diamond-free, the open neighborhood of any vertex is a union of cliques.

\begin{theorem}\label{thm:K_5-free_diamond-free}
There is a $2$-approximation algorithm for \cVD{} in the class of $K_5$-free graphs, and in the class of diamond-free graphs.
\end{theorem}

We note that Theorem~\ref{thm:K_5-free_diamond-free} can be seen as a generalization of the fact that there is a $2$-approximation for \cVD{} in triangle-free graphs, a result that was used by You, Wang, and Cao~\cite{YouEtAl} in their $5/2$-approximation algorithm for (unweighted) \cVD{}.

\section*{Acknowledgments}
We thank the anonymous referees for their careful reading of the paper and helpful remarks.

\bibliographystyle{amsplain}
\bibliography{bibliography}

\end{document}